\documentclass[article]{autart}

\usepackage{wrapfig}
\usepackage{bbold}
\usepackage{enumerate}
\usepackage{cite}
\usepackage{color,xcolor}
\usepackage{amsmath}
\usepackage{amssymb}
\usepackage{graphicx}
\usepackage{algorithm}
\usepackage{algpseudocode}
\usepackage{amsmath}
\usepackage{bbm}
\usepackage{graphics}
\usepackage{subfigure}
\usepackage{epsfig}
\usepackage{multicol} 
\usepackage{amsmath}
\newtheorem{theorem}{Theorem}{}
{}
\newtheorem{lemma}{Lemma}{}
\newtheorem{assumption}{Assumption}
\newtheorem{corollary}{Corollary}{}
\newtheorem{remark}{Remark}{}
\newenvironment{proof}{\hspace{0ex}\textsc{Proof}.\hspace{1ex}}{\hfill$\Box$\newline}

\begin{document}

\begin{frontmatter}

\title{\textcolor{black}{Consensus-Based Distributed Filtering with Fusion Step Analysis}\thanksref{mytitle}}
\thanks[mytitle]{The work by J. Qian and Z. Duan is supported by the National Key R\&D Program of China under Grant 2018AAA0102703, and the National Natural Science Foundation of China under Grants T2121002 and 62173006. The work by P. Duan and L. Shi is supported by a Hong Kong RGC General Research Fund 16210619. \emph{(Corresponding author: Zhisheng Duan.)}}

\author[mymainaddress]{Jiachen Qian}\ead{jcq@pku.edu.cn},
\author[mythirdaddress]{Peihu Duan}\ead{duanpeihu@pku.edu.cn},    
\author[mymainaddress]{Zhisheng~Duan}\ead{duanzs@pku.edu.cn},               
\author[mysecondaryaddress]{Guanrong Chen}\ead{eegchen@cityu.edu.hk},  
\author[mythirdaddress]{Ling Shi}\ead{eesling@ust.hk}  

\address[mymainaddress]{State Key Laboratory for Turbulence and Complex Systems, Department of Mechanics and Engineering Science, College of Engineering, Peking University, Beijing 100871, China}
\address[mysecondaryaddress]{Department of Electrical Engineering, City University of Hong Kong, Hong Kong SAR, China}
\address[mythirdaddress]{Department of Electronic and Computer Engineering, the Hong Kong University of Science and Technology, Clear Water Bay, Kowloon, Hong Kong, China}

\begin{keyword}
Distributed filtering, Consensus, Information fusion, Algebraic Riccati Equation
\end{keyword}

\begin{abstract}
For consensus on measurement-based distributed filtering (CMDF),  through infinite consensus fusion operations during each sampling interval, each node in the sensor network can achieve optimal filtering performance with centralized filtering. However, due to the limited communication resources in physical systems,  \textcolor{black}{the number of fusion steps} cannot be infinite. To deal with this issue, {the present} paper analyzes the performance of CMDF with finite consensus fusion operations. First, by introducing a modified discrete-time algebraic Riccati equation and several novel techniques, the convergence of the estimation error covariance matrix of each sensor is guaranteed \textcolor{black}{under a collective observability condition}. In particular, \textcolor{black}{the steady-state covariance matrix} can be simplified as the solution to a discrete-time Lyapunov equation. Moreover, the performance degradation induced by reduced fusion frequency  is {obtained} in closed form, which establishes an analytical relation between the performance of the CMDF with finite fusion steps and that of centralized filtering. Meanwhile, it provides a trade-off between the filtering performance and the communication cost.  Furthermore, it is shown that the steady-state estimation error covariance matrix  exponentially converges to the centralized optimal steady-state {matrix} with fusion operations tending to infinity {during each sampling interval}. Finally, the theoretical results are verified with illustrative numerical experiments.
\end{abstract}

\end{frontmatter}

\section{Introduction}
In the past {three} decades, the development of wireless sensor network (WSN) facilitates a wide range of engineering applications, e.g., environmental monitoring \cite{Silva20151099} and spacecraft navigation\cite{Vu2015}. {In these scenarios, sensors measure partial states of the {target system} and reconstruct the state of the system}, which can be used to {execute various kinds of tasks}. 

\textcolor{black}{For WSNs consisting of multiple sensors}, frameworks of state estimation are divided into {two main categories}: centralized state estimation and distributed
state estimation. For centralized state estimation, a global fusion center receives and processes {the measurements} from sensors in the WSN. For distributed state estimation, each sensor acts as a local fusion center that collects observation information from its neighbors to reconstruct the state {of the whole system}. Compared with centralized state estimation, distributed state estimation {is easier to implement} and more robust, {while its filtering mechanism} is more challenging to design.

With the development of multi-agent consensus theory, {some typical} consensus-based distributed filtering algorithms were {proposed} and extensively studied \cite{olfati2007distributed,olfati2009kalman,cattivelli2010diffusion,kamal2013information,battistelli2014kullback,Battistelli2015Linear,Wang20181300,he2020distributed,duan2020distributed,li2020Bound,Sayed20219353995}. {When equipped} with the average consensus technique, each sensor first attempts to collect the global information and then  altogether achieve the centralized optimal filtering performance. To name a few approaches, Olfati-Saber \cite{olfati2007distributed} {first formulated a framework of measurement-based distributed filtering (CMDF)}, where each sensor obtains the average of global measurement at every sampling instant by exchanging information with neighbors. Later, Kamal {\it et al.} \cite{kamal2013information} and Battistelli {\it et al.} \cite{battistelli2014kullback,Battistelli2015Linear} proposed some distributed information fusion techniques such as consensus on information matrix, hybrid consensus on measurement and information matrix. \textcolor{black}{Cattivelli {\it et al.} \cite{cattivelli2010diffusion} and Talebi {\it et al.} \cite{Sayed20219353995} developed several distributed filtering structures based on average fusion with neighboring nodes, with the convergence of the estimation error covariance guaranteed.}

{In the above references}, the number of fusion step $L$ between two successive sampling instants $k$ and $k+1$ has a noticeable effect on the performance of the proposed algorithm. \textcolor{black}{Kamal {\it et al.} \cite{kamal2013information}, Kamagarpour and Tomlin \cite{Kamagarpour4738989} and Battilotti {\it et al.} \cite{BATTILOTTI2021109589} pointed out that each sensor node in a distributed sensor network can achieve perfect consensus on the observation information as the number of fusion step $L$ tends to infinity. Consequently, the performance of the distributed filtering algorithm converges to the centralized optimal steady-state performance \cite{Andersonoptimal}.}  However, the infinite fusion step $L$ is impossible to realize in practice,  {which renders the limited applicability of the consensus-based distributed filtering algorithm}. To deal with this issue, Battistelli {\it et al.} \cite{Battistelli2015Linear} and Li {\it et al.} \cite{li2020Bound,liCMweight} analyzed the CMDF algorithm with finite fusion step $L$. {Specifically}, Battistelli {\it et al.} \cite{Battistelli2015Linear} pointed out that the consensus weight takes an integral part in the performance of the {CMDF algorithm}. Li {\it et al.} \cite{li2020Bound,liCMweight} {carried out stability analysis of the {CMDF algorithm} and performed optimization to obtain proper consensus weights}. \textcolor{black}{Battilotti {\it et al.} \cite{BATTILOTTI2021109589} proposed a consensus-based distributed algorithm, and showed the convergence of the steady-state performance to the centralized one with the increase of fusion step $L$. However, in the field of  distributed filtering, the performance analysis was mainly on the boundedness and convergence of the matrix iteration, while {more significant} performance indexes, such as convergence properties of steady-state performance gap with the centralized filtering as $L$ tends to infinity, still require {further exploration}.} 

{In this paper, an in-depth investigation on the performance of the {CMDF algorithm} is carried out, along with an analysis on the number of fusion frequency $L$}. The main contributions of this paper are summarized as follows:
\begin{enumerate}
		\item The convergence condition of the {CMDF} algorithm with a finite fusion step $L$ is derived ({\bf Theorem \ref{thmconverge}}). {It is shown} that as long as the pair of state transition matrix and observation matrix is collectively observable, the steady-state performance of the estimation error covariance matrix can be simplified as the solution to a {discrete-time Lyapunov equation (DLE)}. This result quantitatively describes the steady-state performance degradation induced by insufficient information fusion, {which was rarely analyzed} in the literature \cite{olfati2007distributed,Kamagarpour4738989,li2020Bound}. 
	\item The effect of fusion step $L$ on the steady-state performance of CMDF is discussed. {The} infinite series expression of the difference between two {discrete-time algebraic Riccati equation (DARE)} solutions is formulated. \textcolor{black}{Based on this, {it is shown} that the convergence rate of the performance gap between the centralized optimal case and the distributed case is not slower than exponential convergence, with $L$ tending to infinity ({\bf Theorem \ref{thmric},\ref{thmcov}}).} \textcolor{black}{This result is a theoretical completion of performance analysis of {CMDF}, which reflects the trade-off between estimation accuracy and communication cost {in the distributed setting}.}
	\item Some new properties of the DARE are derived, including the expression of the difference between solutions to two different DARE and the uniform property of solutions to a group of DARE. {These novel properties {play an essential role} in facilitating the applicability of the CMDF algorithm}.
\end{enumerate}    

The remainder of this paper is organized as follows. {Some preliminaries}, including useful lemmas and the problem formulation, are presented in Section \ref{sec2}. The main results, including the derivation and analysis of the steady-state performance of the CM-based algorithm, are presented in Section \ref{sec3}. Some illustrative numerical simulations are presented in Section \ref{sec4}. Conclusions are {drawn} in Section \ref{sec5}. 

\textit{Notation:} For two symmetric matrices $X_{1}$ and $X_{2}$, $X_{1}> X_{2}\left(X_{1}\ge X_{2}\right)$ means $X_{1}-X_{2}$ is positive definite (positive semi-definite). {$exp\left(\cdot\right)$} denotes the exponential function. $\big|a\big|$ denotes the absolute value of real number $a$ or the norm of complex number $a$.  $\mathcal{L}\rhd 0\left(\unrhd 0\right)$ means that \textcolor{black}{all the elements} of matrix $\mathcal{L}$ are positive (non-negative). $\mathbb{E}\left\lbrace x\right\rbrace$ denotes the expectation of a random variable $x$. $\lambda\left(A\right)$ denotes the eigenvalue of matrix $A$. $\rho\left(A\right)$ denotes the spectral radius of $A$.  $\left\|A\right\|_2$ denotes the 2-norm (the largest singular value) of matrix $A$.

\section{Preliminaries}\label{sec2}
\subsection{Graph Theory}
The topology of {a} sensor network is denoted by a graph $\mathcal{G}=\left(\mathcal{V},\mathcal{E},\mathcal{L}\right)$, where $\mathcal{V}=\left\lbrace1,2,\dots,N\right\rbrace$ is the node set, $\mathcal{E}\subseteq \mathcal{V}\times\mathcal{V}$ is the edge set, and $\mathcal{L}=\left[l_{ij}\right]$ is the adjacency matrix of the network. The adjacency matrix reflects the communication among the nodes, $l_{ij}>0\Leftrightarrow\left(i,j\right)\in\mathcal{E}$, which means that sensor $i$ can receive the information from sensor $j$, thus sensor $j$ is denoted as the in-neighbor of sensor $i$, sensor $i$ is denoted as the out-neighbor of sensor $j$. $\mathcal{N}_{i}$ denotes the in-neighbor set of sensor $i$, and $l_{i}$ denotes the $i$-th row of $\mathcal{L}$, which contains the information of $\mathcal{N}_{i}$. \textcolor{black}{In addition, the matrix $\mathcal{L}$ is doubly stochastic, i.e, $\sum_{j=1}^{N}l_{ij}=\sum_{j=1}^{N}l_{ji}=1,\forall i\in\mathcal{V}$, {which ensures the average weighting} of the information from the neighboring nodes}. {The diameter $d$ of graph $\mathcal{G}$ is the length of the longest path between two nodes in the graph.}
  
\subsection{{Target System} and Sensor Model}
{Consider} a network of $N$ sensors, which measure and estimate the states of a target system, described as
\begin{equation}
\begin{aligned}
&x_{k+1}=Ax_k + \omega_k,\quad k = 0,1,2,\ldots\\
&y_{i,k}=C_{i}x_k + v_{i,k},\quad i = 1,2,\ldots,N,
\end{aligned}
\end{equation}
where $x_k\in\mathbb{R}^n$ is the state vector of the system and $y_{i,k}\in\mathbb{R}^{n_i}$ is the measurement vector of sensor $i$,  $\omega_k\in\mathbb{R}^n$ is the process noise with covariance matrix $Q\in\mathbb{R}^{n\times n}$ and $v_{i,k}\in\mathbb{R}^m$ is the observation noise with covariance matrix  $R_{i}\in\mathbb{R}^{n_i \times n_i}$. The sequences $\left\{\omega_k\right\}^{\infty}_{k=0}$ and $\left\{v_{i,k}\right\}^{\infty,N}_{k=0,i=1}$ are \textcolor{black}{mutually uncorrelated white Gaussian noise}. $Q$ and $R_{i}$ are positive definite. Besides, $A$ is the state-transition matrix and $C_{i}$ is the observation matrix of sensor $i$. {For simplicity, use $i\in\mathcal{V}$ to represent the $i$-th sensor of the network.} $C=\left[C_1^T,C_2^T,\dots,C_N^T\right]^T$ is the observation matrix and $R=diag\left(R_{1},\dots,R_{N}\right)$ is the covariance matrix of observation noise of the whole network. $l_{ij}^{(L)}$ is the $(i,j)$-th element of matrix $\mathcal{L}^L$,  {which is the network adjacency matrix at step $L$.}

\subsection{Some Useful Lemmas}
\textcolor{black}{\begin{lemma}\label{lm1}
	If the communication topology corresponding to $\mathcal{L}$ is strongly connected and $\mathcal{L}$ is doubly stochastic with positive diagonal elements, then the matrix $\mathcal{L}^k$ converges to $\frac{1}{N}\textbf{1}\textbf{1}^T$ exponentially as $k\to\infty$.  
\end{lemma}
\vspace{6pt}
\begin{proof}
	The strongly connected property guarantees that the matrix $\mathcal{L}$ is irreducible. With Theorem 8.5.2 and Lemma 8.5.5 in \cite{Horn1985}, the matrix $\mathcal{L}$ is also primitive with positive diagonal elements. Thus, there is only one eigenvalue of $\mathcal{L}$ equal to the spectral radius 1, and the norms of all other eigenvalues are strictly less than one. Consider any eigenvector corresponding to an eigenvalue $\big|\lambda\big|<1$, denoted as ${\it\bf x}$. Then, one has
	$$
	{\bf 1}^T\mathcal{L}{\bf x}={\bf 1}^T{\bf x}=\lambda\textbf{1}^T\textbf{x},
	$$
	where the first equality follows from the fact that $\mathcal{L}$ is doubly stochastic. Therefore one has $\textbf{1}^T\textbf{x}=0$. Thus, all the eigenvalues of $\mathcal{L}$ except 1 are also eigenvalues of $\mathcal{L}-\frac{1}{N}\textbf{1}\textbf{1}^T$. With Theorem 1 in \cite{Xiao200465}, the matrix $\mathcal{L}^k$ converges to $\frac{1}{N}\textbf{1}\textbf{1}^T$ exponentially with the increase of $k$, and the convergence rate is not slower than the norm of the second largest eigenvalue of $\mathcal{L}$.
\end{proof}}
\vspace{6pt}
\begin{lemma}\label{lm2}
	(Matrix {Inversion} Lemma)
	For any matrix $P,Q,C$ of proper dimensions, if $P^{-1}$ and $Q^{-1}$ exist, then the following {equality} holds:
	$$
	\left(P^{-1}+C^TQ^{-1}C\right)^{-1}=P - PC^T\left(CPC^T+Q\right)^{-1}CP.
	$$
\end{lemma}
\vspace{6pt}
\begin{lemma}\label{lm3}
  {For any matrix $A \in \mathbb{R}^{n \times n}$}, the following inequality holds: 
	$$
	\left\|A^k\right\|_2\leq\sqrt{n}\sum_{j=0}^{n-1}\binom{n-1}{j}\binom{k}{j}\big\|A\big\|_2^j\rho\left(A\right)^{k-j},
	$$ 
	\textcolor{black}{where $\binom{m}{n}$ is the combinatorial number, $\binom{k}{j}=0$ for $j>k$, and $\rho\left(A\right)$ and $\big\|A\big\|_2$ denote the spectral radius and maximum singular value of $A$, respectively}.
\end{lemma}

The proof is given in Appendix \ref{Pflm3}.

\subsection{Problem Formulation}

{In this paper}, the fundamental problem of CM-based distributed filtering algorithm with finite fusion step $L$ is formulated {as follows:}
\begin{enumerate}
	\item {Evaluate} the effect induced by insufficient fusion on the steady-state performance of the CM-based distributed filtering algorithm.
	\item Based on the above analysis, {find out} the relationship between fusion step $L$ and performance degradation compared with the centralized optimal performance.
\end{enumerate}

\section{Main Results}\label{sec3}
In this section, {first,} the convergence property of the proposed CMDF algorithm is analyzed, including the iteration of parameter matrix and error covariance matrix. Then, the steady-state performance between the CM-based filter and the centralized filter are compared. {To do so, the following two assumptions are needed}.
\begin{assumption}\label{as1}
	The communication topology is strongly connected.
\end{assumption}
\begin{assumption}\label{as2}
	The system pair $\left(A,C\right)$ is collectively observable.
\end{assumption}

\textcolor{black}{As proposed in \cite{battistelli2014kullback, kamal2013information, battistelli2018distributed, he2020distributed}, the above two assumptions are both mild for distributed filtering problems. Generally speaking, Assumption \ref{as1} is to ensure globally average weighting of the information from the nodes of the sensor network so as to achieve the performance of  the centralized optimal case, and Assumption \ref{as2} is essential for the stability of the filtering algorithm.}
\subsection{Algorithm Convergence Analysis}
In this subsection, consider the following CM-based distributed Kalman filtering algorithm {proposed in} \cite{olfati2007distributed,Kamagarpour4738989,Battistelli2015Linear}, as shown in Algorithm \ref{alg1}, {where $N$ is the {\it a priori} parameter shared by all nodes in the sensor network, and the doubly stochastic matrix $\mathcal{L}$ can be shared as global parameters or obatined in a distributed way when the communication graph is strongly connected \cite{GHARESIFARD2012539}.}
\begin{algorithm}
	\caption{Distributed Information Fusion Algorithm}
	\label{alg1}
	\textbf{Input:}\\
	$\hat{x}_{i,0|0}, P_{i,0|0}\quad i=1,2,\dots,N$\\
	\textbf{Prediction:}\\
	$\hat{x}_{i,k|k-1}=A\hat{x}_{i,k-1|k-1}\\ P_{i,k|k-1}=AP_{i,k-1|k-1}A^T+Q$\\
	\textbf{Fusion:}\\
	Set 
	$$
		S_{i,k}^{(0)}=NC_i^TR_i^{-1}C_i, \qquad I_{i,k}^{(0)}=NC_i^TR_i^{-1}y_{i,k}
	$$\\
	For $m=1,2,\dots,L$
	$$
	\begin{aligned}
		S_{i,k}^{(m)}=\sum_{j=1}^{N}l_{ij}S_{j,k}^{(m-1)}, \qquad I_{i,k}^{(m)}=\sum_{j=1}^{N}l_{ij}I_{j,k}^{(m-1)}\;\;
	\end{aligned}
	$$
	\textbf{Correction:}\\
	$P_{i,k|k}=\big(P_{i,k|k-1}^{-1}+S_{i,k}^{(L)}\big)^{-1}$\\
	$\hat{x}_{i,k|k}=P_{i,k|k}\big(P_{i,k|k-1}^{-1}\hat{x}_{i,k|k-1}+I_{i,k}^{(L)}\big)$\\
\end{algorithm}

 {With this algorithm, as} the number of fusion step $L$ tends to infinity, each sensor node can precisely obtain the information $\sum_{i=1}^{N}C_i^TR_i^{-1}C_i$ and $\sum_{i=1}^{N}C_i^TR_i^{-1}y_{i,k}$ in a distributed way. Based on this premise, Kamgarpour and Tomlin \cite{Kamagarpour4738989} proved that, with $k$ tending to infinity, the matrix $P_{i,k|k-1}$ converges to the centralized optimal steady-state performance, i.e., the solution $P$ to the DARE 
\begin{equation}\label{dare}
	P=APA^T+Q-APC^T\left(CPC^T+R\right)^{-1}CPA^T.
\end{equation}

However, due to the limitation of the communication rate, it is impossible to communicate with the neighboring sensor node {for infinitely many times} between two sampling instants. Thus, {the objective of this subsection} is to derive the convergence property of the algorithm based on {a finite fusion step $L$.}

With Algorithm \ref{alg1}, the iteration of {$P_{i,k|k}^{-1}$} can be formulated as
\begin{equation}
	P_{i,k|k}^{-1}=P_{i,k|k-1}^{-1}+N\sum_{j=1}^{N}l_{ij}^{(L)}C_{j}^{T}R_{j}^{-1}C_{j}.
\end{equation}
{Let} the modified observation matrix for each sensor {be}
$$
\tilde{C}_i^{(L)}=\Big[sign\big(l_{i1}^{(L)}\big) C_1^T,\cdots,sign\big(l_{iN}^{(L)}\big)C_N^T\Big]^T,
$$
and the modified noise covariance matrix {be}
\textcolor{black}{
$$
\begin{aligned}
	\bar{R}_{i}^{(L)}&=diag\left(sign\big(l_{i1}^{(L)}\big)R_{1},\cdots,sign\big(l_{iN}^{(L)}\big)R_{N}\right)\\
	\tilde{R}_{i}^{(L)}&=diag\left(\frac{1}{Nl_{i1}^{(L)}}R_{1},\cdots,\frac{1}{Nl_{iN}^{(L)}}R_{N}\right),
\end{aligned}
$$
where 
$$
sign\left(x\right)=\left\{\begin{aligned}
	-&1,\quad x<0\\
	&0,\quad x=0\\
	&1,\quad x>0
\end{aligned}\right.,
$$ 
and $\frac{1}{Nl_{i1}^{(L)}}$ is set to $0$ if the denominator $Nl_{ij}^{(L)}=0$.}
As $\mathcal{L}^L$ is a non-negative matrix, i.e., $\mathcal{L}^L\unrhd 0$, the term $sign\big(l_{ij}^{(L)}\big)$ only takes the value of $0$ or $1$. {The following Lemma describes} the convergence property of the iteration $P_{i,k+1|k}$.
\vspace{5pt}
{\begin{lemma}\label{lmric}
	If $\big(A,\tilde{C}_i^{(L)}\big)$ is observable for all $i\in\mathcal{V}$, then $P_{i,k+1|k}$ converges to the solution of the DARE 
	\begin{equation}\label{ricori}
		\begin{aligned}
			P_i^{(L)}=&AP_i^{(L)}A^{T}+Q-AP_i^{(L)}\big(\tilde{C}_i^{(L)}\big)^T\\
			&\times\Big(\tilde{C}_i^{(L)} P_i^{(L)}\big(\tilde{C}_i^{(L)}\big)^T+\tilde{R}_i^{(L)}\Big)^{-1}\tilde{C}_i^{(L)} P_i^{(L)}A^{T},
		\end{aligned}
	\end{equation}
	i.e.,
    $$
     \lim_{k\to \infty}P_{i,k+1|k}=P_{i}^{(L)},\qquad \forall i\in\mathcal{V}.
    $$
\end{lemma}}
\begin{proof}
	{The iteration of $P_{i,k+1|k}$ can be reformulated as}
	
	\begin{small}
		\begin{equation}\label{eqreformu}
			\begin{aligned}
				&P_{i,k+1|k}=A\left(P_{i,k|k-1}^{-1}+\big(\tilde{C}_i^{(L)}\big)^T\big(\tilde{R}_{i}^{(L)}\big)^{-1}\tilde{C}_i^{(L)}\right)^{-1}A^{T}+Q.\\
			\end{aligned}	
		\end{equation}
	\end{small}

{With Lemma \ref{lm2}, the following equality is obtained:}

{
$$
	\begin{aligned}
		&\Big(P_{i,k|k-1}^{-1}+\big(\tilde{C}_i^{(L)}\big)^T\big(\tilde{R}_{i}^{(L)}\big)^{-1}\tilde{C}_i^{(L)}\Big)^{-1}\\
		&=P_{i,k|k-1}-P_{i,k|k-1}\big(\tilde{C}_i^{(L)}\big)^T\\
	&\quad\times\Big(\tilde{C}_i^{(L)} P_{i,k|k-1}\big(\tilde{C}_i^{(L)}\big)^T+\tilde{R}_i^{(L)}\Big)^{-1}\tilde{C}_i^{(L)} P_{i,k|k-1}.
\end{aligned}
$$}
{The equation \eqref{eqreformu} can be reformulated as
$$
\begin{aligned}
	&P_{i,k+1|k}=AP_{i,k|k-1}A^T+Q-AP_{i,k|k-1}\big(\tilde{C}_i^{(L)}\big)^T\\
	&\quad\times\Big(\tilde{C}_i^{(L)} P_{i,k|k-1}\big(\tilde{C}_i^{(L)}\big)^T+\tilde{R}_i^{(L)}\Big)^{-1}\tilde{C}_i^{(L)} P_{i,k|k-1}A^{T}.\\
\end{aligned}
$$}

{Thus, equation \eqref{ricori} and the steady-state form of equation \eqref{eqreformu} are equivalent}. With the result in \cite{Andersonoptimal}, if $\big(A,\tilde{C}_i^{(L)}\big)$ is observable, then $P_{i,k+1|k}$ converges to the solution of the DARE \eqref{ricori}.
\end{proof}
\vspace{6pt}
\begin{remark}
	If the number of fusion step is less than the diameter of the communication topology, i.e., $L<d$, \textcolor{black}{then the matrix $\tilde{R}_i^{(L)}$ is not invertible since some $l_{ij}^{(L)}$ may be $0$. However, {one can eliminate the corresponding zero blocks} in $\tilde{C}_i^{(L)}$ and $\tilde{R}_i^{(L)}$ to make the modified $\tilde{R}_i^{(L)}$ invertible, or replace the inverse sign with generalized inverse, i.e., $\big(\tilde{R}_i^{(L)}\big)^{\dagger}$.  These two kinds of modification do not affect the observability of the pair $\big(A,\tilde{C}_i^{(L)}\big)$ and the equivalent relationship $$\big(\tilde{C}_i^{(L)}\big)^T\big(\tilde{R}_i^{(L)}\big)^{\dagger}\tilde{C}_i^{(L)}=N\sum_{j=1}^{N}l_{ij}^{(L)}C_{j}^{T}R_{j}^{-1}C_{j}.$$ 
   {In order to simplify the notation, $\tilde{C}_i^{(L)}$, $\tilde{R}_i^{(L)}$, $\big(\tilde{R}_i^{(L)}\big)^{-1}$  and $\bar{R}_i^{(L)}$ will be kept to denote the modified observation and noise matrices.}}
\end{remark}
\vspace{6pt}

With the DARE \eqref{ricori}, it is clear that the steady-state performance of the CMDF algorithm is closely related to the number of fusion step $L$ between two sampling instants. Furthermore, the matrix $P_i^{(L)}$ is determined by $A,\tilde{C}_i^{(L)},\tilde{R}_i^{(L)}$, which is apparently not the real state estimation error covariance matrix for sensor $i$. Therefore, in addition to the above analysis on $P_{i,k+1|k}$, {it also needs to analyze} the steady-state performance of the real error covariance matrix of each sensor $i$.

Denote the estimation error of sensor node $i$ as $\tilde{x}_{i,k|k-1}=x_k-\hat{x}_{i,k|k-1}$, $\tilde{x}_{i,k|k}=x_{k}-\hat{x}_{i,k|k}$, and denote the estimation error covariance matrix as $\tilde{P}_{i,k|k-1}=\mathbb{E}\big\lbrace\tilde{x}_{i,k|k-1}\tilde{x}_{i,k|k-1}^{T}\big\rbrace$, $\tilde{P}_{i,k|k}=\mathbb{E}\big\lbrace\tilde{x}_{i,k|k}\tilde{x}_{i,k|k}^{T}\big\rbrace$.
{
\begin{theorem}\label{thmconverge}
	If $\big(A,\tilde{C}_i^{(L)}\big)$ is observable for $i\in\mathcal{V}$, then the iteration of $\tilde{P}_{i,k|k-1}$ will converge to the solution to the discrete-time Lyapunov equation (DLE)
	\begin{equation}\label{lyaeq}
		\begin{aligned}
			&\tilde{P}_i^{(L)}=\tilde{A}_{P_i^{(L)}}\tilde{P}_i^{(L)}\tilde{A}_{P_i^{(L)}}^T+Q+K_{P_i^{(L)}}\bar{R}_iK_{P_i^{(L)}}^T,\\
		\end{aligned}
	\end{equation}
	where
		$$
		\begin{aligned}
			\tilde{A}_{P_i^{(L)}}\triangleq A-&AP_i^{(L)}\big(\tilde{C}_i^{(L)}\big)^{T}\\
			&\times\Big(\tilde{C}_i^{(L)}\tilde{P}_i^{(L)}\big(\tilde{C}_i^{(L)}\big)^T+\tilde{R}_i^{(L)}\Big)^{-1}\tilde{C}_i^{(L)}
		\end{aligned}
		$$ 
		$$
		\begin{aligned}
			\quad\; K_{P_i^{(L)}}\triangleq\; &AP_{i}^{(L)}\big(\tilde{C}_i^{(L)}\big)^T\\
			&\quad\;\times\Big(\tilde{C}_i^{(L)}P_{i}^{(L)}\big(\tilde{C}_i^{(L)}\big)^T+\tilde{R}_{i}^{(L)}\Big)^{-1}.\;\qquad\quad
		\end{aligned}
		$$
\end{theorem}
\vspace{6pt}
	The proof is given in Appendix \ref{Pfconverge}.}
\vspace{6pt}
\begin{remark}
    For CMDF, {the explicit values of 
    $\sum_{i=1}^{N} C_{i}^{T}\cdot$ $R_{i}^{-1}y_{i}$ and $\sum_{i=1}^{N}C_{i}^TR_{i}^{-1}C_{i}$} in each sampling instant cannot be obtained by each node with finite fusion step $L$.  
    Compared with {the centralized optimal setting}, the effect of imprecise information on the state estimation performance will accumulate, and \textcolor{black}{the corresponding degraded estimation performance will become steady}. A more in-depth discussion on the performance degradation induced by finite  fusion step $L$ {is proposed in Theorem \ref{thmric} and \ref{thmcov}}.
    
    \textcolor{black}{In addition to these, Theorem 1 also provides the least requirement on the stability of the CM-based distributed filtering algorithm. Denote the L-step neighbor of node $i$ as the nodes that can reach node i in $L$ steps. Lemma 4 and Theorem 1 indicate that the lowest bound of $L$ for the boundedness of $\tilde{P}_{i,k|k-1}, \forall i \in \mathcal{V}$, is that each pair $(A,\tilde{C}_i^{(L)})$ is observable for all $i\in\mathcal{V}$, namely, for each node $i\in\mathcal{V}$, its L-step neighbors need to be collectively observable. Note that if the fusion is not sufficient enough, i.e., the pairs $\big(A,\tilde{C}_i^{(L)}\big)$ are not observable for some $i \in \mathcal{V}$, then the real covariance matrix $\tilde{P}_{i,k+1|k}^{(L)}$ may diverge for an unstable $A$. Meanwhile, for not very large $L$, possible inconsistency of the CM-based algorithm may be induced by the overconfidence on the observation information $y_{j,k}$ with $Nl_{ij}^{(L)}$ larger than 1. This may further lead to a relatively large gap between $P_{i}^{(L)}$ and $\tilde{P}_i^{(L)}$ and affect the steady-state performance of the filtering algorithm. However, the CM-based distributed filtering algorithm can achieve the centralized optimal performance with any degree of precision for sufficiently large $L$, as shown below.}
\end{remark}


\textcolor{black}{Now, one is ready to complete the analysis about the convergence of $P_{i,k|k-1}$ and $\tilde{P}_{i,k|k-1}$ to the centralized performance,} and to complete {the formulation of} the steady-state performance of iterations as the solutions to DARE (\ref{ricori}) and DLE (\ref{lyaeq}), respectively.

When $L$ is larger than the diameter of the communication graph, i.e., $L\ge d$, one has $l_{ij}^{(L)}>0$ for all $i,j$ \cite{battistelli2014kullback}. {Thus, one can} rewrite the DARE (\ref{ricori}) for each sensor node as
$$
	\begin{aligned}
		P_i^{(L)}=&AP_i^{(L)}A^{T}+Q\\
		&-AP_i^{(L)}C^{T}\big(C P_i^{(L)}C^T+\tilde{R}_i^{(L)}\big)^{-1}C P_i^{(L)}A^{T}.\\
	\end{aligned}
$$
It is obvious that, \textcolor{black}{compared to the centralized optimal setting}, the solution of the matrix $P_{i}^{(L)}$ is based on a mismatched $\tilde{R}_i^{(L)}$ instead of $R$. Therefore, the effect of non-sufficient fusion is formulated as the mismatch of noise covariance matrix $R$. In the next subsection, the effect of mismatched $R$ on the properties of $P_{i}^{(L)}$ and $\tilde{P}_{i}^{(L)}$ with {$L\ge d$  will be further discussed.}
\subsection{Performance Analysis}
Before discussing the comparison of filtering performance, the following {two Lemmas are established} to illustrate the properties of solutions to DARE. 
{For simplicity}, use the term $dare\left(A,C,Q,R\right)$ to denote the solution $P$ to DARE (\ref{dare}). 
\vspace{6pt}
\begin{lemma}\label{lmmono}
	{The solution $P$ to the DARE \eqref{dare} is monotonically increasing with $R$}, i.e., if $R_1\ge R_2>0$, $P_1=dare\left(A,C,Q,R_1\right)$ and $P_2=dare\left(A,C,Q,R_2\right)$, then $P_1\ge P_2$.
\end{lemma}
\vspace{6pt}
The proof of this Lemma is similar to the proof of Theorem 1 in \cite{zhisheng2012effects}, thus is omitted.

{Similarly, denote $\tilde{A}_{P}=A-APC^T\left(CPC^T+R\right)^{-1}C$}, where $P$ is the solution to DARE \eqref{dare}.
\vspace{5pt}
{\begin{lemma}\label{lmstable}
     The matrix $\tilde{A}_P$ is Schur stable and 
	 $$
	 \begin{aligned}
		\rho\big(\tilde{A}_P\big)\leq \sqrt{1-\frac{\lambda_{min}\left(Q\right)}{\lambda_{max}\left(P\right)}}\;,\quad
		\big\|\tilde{A}_P\big\|_2\leq \sqrt{\frac{\lambda_{max}\left(P\right)}{\lambda_{min}\left(Q\right)}}.
	 \end{aligned}
	 $$ 
\end{lemma}
\vspace{6pt}
	The proof is given in Appendix \ref{Pflmstable}.
}
\vspace{6pt}
\begin{remark}
	 Lemma \ref{lmstable} shows the property of the feedback matrix $\tilde{A}_{P}$, {in terms of} spectral radius and singular values. This Lemma is of vital importance to analyze the uniform property of the solutions to a group of Riccati Equations, {as can be seen} in the proof of Theorem \ref{thmric}.
\end{remark}
\vspace{6pt}

In order to discuss the performance degradation due to the effect of insufficient fusion, {it needs to compare} the difference between two DARE first.
{Introduce two Riccati Equations with different $R$:}
\begin{equation}\label{twodare}
	\begin{aligned}
		P_1 = AP_1A^T+Q-AP_1C^T\left(CP_1C^T+R_1\right)^{-1}CP_1A^T\;\\
		P_2=AP_2A^T+Q-AP_2C^T\left(CP_2C^T+R_2\right)^{-1}CP_2A^T.
	\end{aligned}
\end{equation}

\begin{lemma}\label{lmdifference}
	{Suppose $P_1$ and $P_2$ satisfy the DARE \eqref{twodare} respectively.} Then, one has
	$$
	P_1-P_2=\sum_{k=0}^{\infty}\tilde{A}_{P_1}^kA\bar{P}_1C^T\left(R_2^{-1}-R_1^{-1}\right)C\bar{P}_2A^T\big(\tilde{A}_{P_2}^T\big)^k,
	$$
	where 
	$$
	\begin{aligned}
	&\bar{P}_1=\big(P_1^{-1}+C^TR_1^{-1}C\big)^{-1},\quad \bar{P}_2=\big(P_2^{-1}+C^TR_2^{-1}C\big)^{-1}\\
	&\tilde{A}_{P_1}=A-K_{P_1}C,\qquad K_{P_1}=AP_1C^T\left(CP_1C^T+R_1\right)^{-1}\\
	&\tilde{A}_{P_2}=A-K_{P_2}C,\qquad K_{P_2}=AP_2C^T\left(CP_2C^T+R_2\right)^{-1}.
	\end{aligned}
	$$
\end{lemma}
\vspace{6pt}
	The proof is given in Appendix \ref{Pflmdifference}.

With Lemma \ref{lmdifference}, {one can estimate} the difference of two solutions to DARE with different $R$ by a sum of infinite series:

$$
\begin{aligned}
\big\|P_1-P_2\big\|_2\leq&\left\|\bar{P}_1\right\|_2\left\|\bar{P}_2\right\|_2\big\|C\big\|_2^2\big\|R_2^{-1}-R_1^{-1}\big\|_2\\ &\times\big\|A\big\|_2^2\sum_{k=0}^{\infty}\big\|\tilde{A}_{P_1}^k\big\|_2\big\|\tilde{A}_{P_2}^k\big\|_2.
\end{aligned}
$$

	Let $R_1=R$, and $R_2=\tilde{R}_{i}^{(L)}\left(L\ge d\right)$. Then, an estimation of the performance degradation compared with {the centralized setting is given by}
    $$
	\begin{aligned}
		\big\|P-P_i^{(L)}\big\|_2\leq&\left\|\bar{P}\right\|_2\big\|\bar{P}_i^{(L)}\big\|_2\big\|C\big\|_2^2\;\big\|\big(\tilde{R}_i^{(L)}\big)^{-1}-R^{-1}\big\|_2\\
		&\times\big\|A\big\|_2^2\sum_{k=0}^{\infty}\big\|\tilde{A}_{P}^k\big\|_2\big\|\tilde{A}_{P_i^{(L)}}^k\big\|_2.
	\end{aligned}
	$$

	{Next, it remains} to analyze the property of the infinite sum $\sum_{k=0}^{\infty}\big\|\tilde{A}_{P}^k\big\|_2\big\|\tilde{A}_{P_i^{(L)}}^k\big\|_2$ to obtain the relationship between $\big\|P-P_i^{(L)}\big\|_2$ and $\big\|\big(\tilde{R}_i^{(L)}\big)^{-1}-R^{-1}\big\|_2$, where Lemma 3 and Lemma \ref{lmstable} will be utilized.

	As the matrices $\tilde{A}_{P}$ and $\tilde{A}_{P_i^{(L)}}$ are Schur stable, one can find a bound of the infinite sum $\sum_{k=0}^{\infty}\big\|\tilde{A}_{P}^k\big\|_2\big\|\tilde{A}_{P_i^{(L)}}^k\big\|_2$ and obtain an approximate result as 
	$$
	\big\|P-P_i^{(L)}\big\|_2\leq M\big\|\big(\tilde{R}_i^{(L)}\big)^{-1}-R^{-1}\big\|_2,
	$$
	where $M$ is a constant. However, for different $\tilde{R}_i^{(L)}$, {the constant $M$} may also be different. {This issue makes the approximately linear relationship between $\big\|P-P_i^{(L)}\big\|_2$ and $\big\|\big(\tilde{R}_i^{(L)}\big)^{-1}-R^{-1}\big\|_2$ {not trivial to obtain, meanwhile it brings about} difficulties to the analysis of the uniform property of $\big\|P-P_i^{(L)}\big\|_2$ for all $i\in\mathcal{V}$ and $L\ge d$.} {Therefore, it needs to} discuss the uniform property of the infinite sum $\sum_{k=0}^{\infty}\big\|\tilde{A}_{P}^k\big\|_2\big\|\tilde{A}_{P_i^{(L)}}^k\big\|_2$ for a group of $\tilde{R}_i^{(L)}$ to {reveal} the trade-off between the number of fusion step $L$ and the estimation performance.
	\vspace{6pt}
{\begin{theorem}\label{thmric}
	For any given $\left(A,C\right)$ and $Q,R$, {there exist two constants} $M_1>0$, $0<q_1<1$, such that
	$$
	\big\|P-P_i^{\left(L\right)}\big\|_2\leq M_1q_1^L, \quad\forall i\in\mathcal{V},\forall L\ge d.
	$$
\end{theorem}
\vspace{6pt}
	The proof is given in Appendix \ref{Pfthmric}.}

\vspace{6pt}
\begin{remark}
	The main difficulty in proving Theorem \ref{thmric} is to formulate a uniform upper bound in terms of $\sum_{k=0}^{\infty}\big\|\tilde{A}_{P}^k\big\|_2\big\|\tilde{A}_{P_i^{(L)}}^k\big\|_2$ for all $L\ge d$ and $i\in\mathcal{V}$. {As one can use an exponentially convergent bound for} $\big\|P_{i}^{(L)}-P\big\|_2$, the convergence speed of $\big\|P_{i}^{(L)}-P\big\|_2$ with increasing $L$ is not slower than exponential convergence.
	Theorem \ref{thmric} also illustrates that, with the number of fusion step tending to infinity, the steady-state performance matrix $P_{i}^{\left(L\right)}$ will exponentially converge to the centralized optimal steady-state performance, which reflects the trade-off between estimation performance and communication cost among sensor nodes {in a distributed setting.}
\end{remark}
\vspace{6pt}
\begin{remark}
	It is worth mentioning that the results proposed in \cite{olfati2007distributed,Kamagarpour4738989,kamal2013information,BATTILOTTI2021109589} only discussed the asymptotic optimal property of the corresponding algorithm, i.e., the case when $L\to\infty$. However, {the case with finite iteration step $L$ was rarely discussed} in the literature. Meanwhile, the references \cite{battistelli2014kullback,he2018consistent,talebi2019distributed,talebi2020quaternion} discussed the performance of distributed information fusion algorithms with finite fusion steps but did not compare the filtering performance with the centralized optimal setting. Thus, Theorem \ref{thmric} {not only characterizes} the property of the filtering algorithm with insufficient information fusion but also {compares the performance degradation with the centralized setting}, which is a more in-depth {result} in the area of distributed filtering. 
\end{remark}
\vspace{6pt}

It is apparent that the matrix $P_i^{(L)}$ is not the real error covariance matrix, i.e., 
$$
P_i^{(L)}\neq\lim_{k\to\infty}\mathbb{E}\big\lbrace\tilde{x}_{i,k|k-1}\tilde{x}_{i,k|k-1}^T\big\rbrace.
$$ 
With the aforementioned analysis, the iteration of the real covariance $\mathbb{E}\big\lbrace\tilde{x}_{i,k|k-1}\tilde{x}_{i,k|k-1}^T\big\rbrace$ will finally converge to a solution to DLE. Thus, with the technique developed above, {one can further} discuss the real performance degradation of $\mathbb{E}\big\lbrace\tilde{x}_{i,k|k-1}\tilde{x}_{i,k|k-1}^T\big\rbrace$ due to insufficient fusion $L$.

{First, the aim is to} analyze the distance between the matrix $P_{i}^{(L)}$ and the matrix $\tilde{P}_i^{(L)}$. 
\vspace{6pt}
{\begin{theorem}\label{thmcov}
	For any given $\left(A,C\right)$ and $Q,R$, there exist two constants $M_2>0$, $0<q_2<1$, such that
	$$
	\big\|\tilde{P}_i^{(L)}-P_i^{\left(L\right)}\big\|_2\leq M_2q_2^L, \quad\forall i\in\mathcal{V},\forall L\ge d.
	$$
\end{theorem}
\vspace{6pt}
	The proof is given in Appendix \ref{Pfthmcov}.
}
\vspace{6pt}

The following Corollaries finally illustrate the performance degradation of $\tilde{P}_{i}^{(L)}$ compared with the centralized steady-state performance $P$ and the relation between the convergence rate $q$ and the spectral property of the weight matrix $\mathcal{L}$.
\vspace{6pt}
\begin{corollary}\label{exponential}
	For any given $\left(A,C\right)$ and $Q,R$, there exist two constants $M_3>0$, $0<q_3<1$, such that
	$$
	\big\|\tilde{P}_i^{(L)}-P\big\|_2\leq M_3q_3^L, \quad\forall i\in\mathcal{V},\forall L\ge d.
	$$
\end{corollary}
\vspace{6pt}
\begin{proof}
	The corollary can be proved with Theorem \ref{thmric}, Theorem \ref{thmcov} and the triangular inequality of 2-norm:
	$$
	\big\|\tilde{P}_i^{(L)}-P\big\|_2\leq \big\|\tilde{P}_i^{(L)}-P_{i}^{(L)}\big\|_2+\big\|P_{i}^{(L)}-P\big\|_2.
	$$
\end{proof}
\textcolor{black}{
\begin{corollary}
	The parameters $q_1,q_2,q_3$ mentioned in Theorem \ref{thmric}, Theorem \ref{thmcov} and Corollary \ref{exponential}  are not larger than the norm of the second largest eigenvalue of the stochastic matrix $\mathcal{L}$.
\end{corollary}
\begin{proof}
	From the proof of Theorem \ref{thmric} and Theorem \ref{thmcov}, the decay rates $q_1,q_2,q_3$ are actually bounded by that of $Nl_{ij}^{(L)}-1$. Thus, the result follows. Note that even if the second largest eigenvalue of $\mathcal{L}$ has generalized eigenvector, the term $q_i,\;i=1,2,3$ can still be chosen as close to the norm of the second largest eigenvalue as possible.
\end{proof}}
\subsection{Discussion}
\textcolor{black}{
	In Theorem \ref{thmcov}, the deviation between $\tilde{P}_i^{(L)}$ and $P_i^{(L)}$ are reformulated in the series form (\ref{sumCov}), which indicates that the deviation is in essence determined by the difference of the noise terms $R$ and $\tilde{R}_i^{(L)}$. With the Kalman equality $K_{P_{i}^{(L)}}=A\bar{P}_i^{(L)}C^T\big(\tilde{R}_{i}^{(L)}\big)^{-1}$ and the uniform boundedness of $\sum_{k=0}^{\infty}\big\|\tilde{A}_{P_i^{(L)}}^{k}\big\|_2^2$, a more inherent result is that the decay rates of $\big\|P_i^{(L)}-\tilde{P}_i^{(L)}\big\|_2$ are uniformly bounded by that of $\tilde{l}_{ij}^{(L)}=\big(Nl_{ij}^{(L)}-1\big)^2+\big(Nl_{ij}^{(L)}-1\big)$, in which the linear term $\big(Nl_{ij}^{(L)}-1\big)$ dominates the convergence speed for sufficiently large $L$. Based on these properties,  the exponential convergence of $\big\|\tilde{P}_{i}^{(L)}-P_{i}^{(L)}\big\|_2$ is finally proved.}

    \textcolor{black}{Corollary 2 also reflects the superiority of the exponential convergence. An intuitive idea from Corollary 2 is that the parameter $l_{ij}$ could be designed to minimize the norm of the second largest eigenvalue of $\mathcal{L}$ or minimize the $L$-step consensus error to reach better performance. Some distributed parameter tuning techniques such as subgradient method \cite{Xiao200465} or graph filtering \cite{yi8678811} can also be utilized, which is an interesting direction for future exploration. } 
    
    \textcolor{black}{To summarize, the contents and proof of Theorem \ref{thmcov} are the basis for the derivation of Corollaries 1 and 2, and Corollary 2 reveals the connection between the convergence rate $q$ and the spectral property of $\mathcal{L}$. }

\section{Simulation}\label{sec4}
In this section, a target tracking numerical experiment is provided to illustrate the effectiveness of the proposed algorithm. The state transition matrix has the expression
$$
\begin{aligned}
	&a_{k}=\begin{pmatrix}
		1&T\\
		0&1
	\end{pmatrix}\qquad
	&A_k=\begin{pmatrix}
	a_{k}&0_{2\times 2}\\
	0_{2\times 2}&a_{k}
	\end{pmatrix}.
\end{aligned}
$$ 
The error covariance matrix $Q$ {takes the form of}
$$
\begin{aligned}
	&G=\begin{pmatrix}
		\frac{T^3}{3}&\frac{T^2}{2}\\
		\frac{T^2}{2}&T
	\end{pmatrix}\quad
	&Q=\begin{pmatrix}
		G&0.5G\\
		0.5G&G\\
		\end{pmatrix},
\end{aligned}
$$
where the sample interval is set to be $T=1$. There are three kinds of sensors with the observation matrix:
$$
\begin{aligned}
	&C^{(1)}=\left[1,0,0,0\right]\\
	&C^{(2)}=\left[0,0,1,0\right]\\
	&C^{(3)}=\left[0,0,0,0\right].
\end{aligned}
$$
The whole network consists of 20 sensor nodes, including 3 sensors of kind $C^{(1)}$, 3 sensors of kind $C^{(2)}$, 14 sensors of kind $C^{(3)}$, and $R_{i}=1,\forall i\in\mathcal{V}$. {The locations of the sensors are randomly set in a $300\times300$ {region} and each sensor is with a communication radius of 130.} The communication topology of sensor network is randomly generated in the numerical experiments, {as presented in Fig.\ref{topology}.}
\begin{figure}
	\centering
	\includegraphics[width=0.4\textwidth]{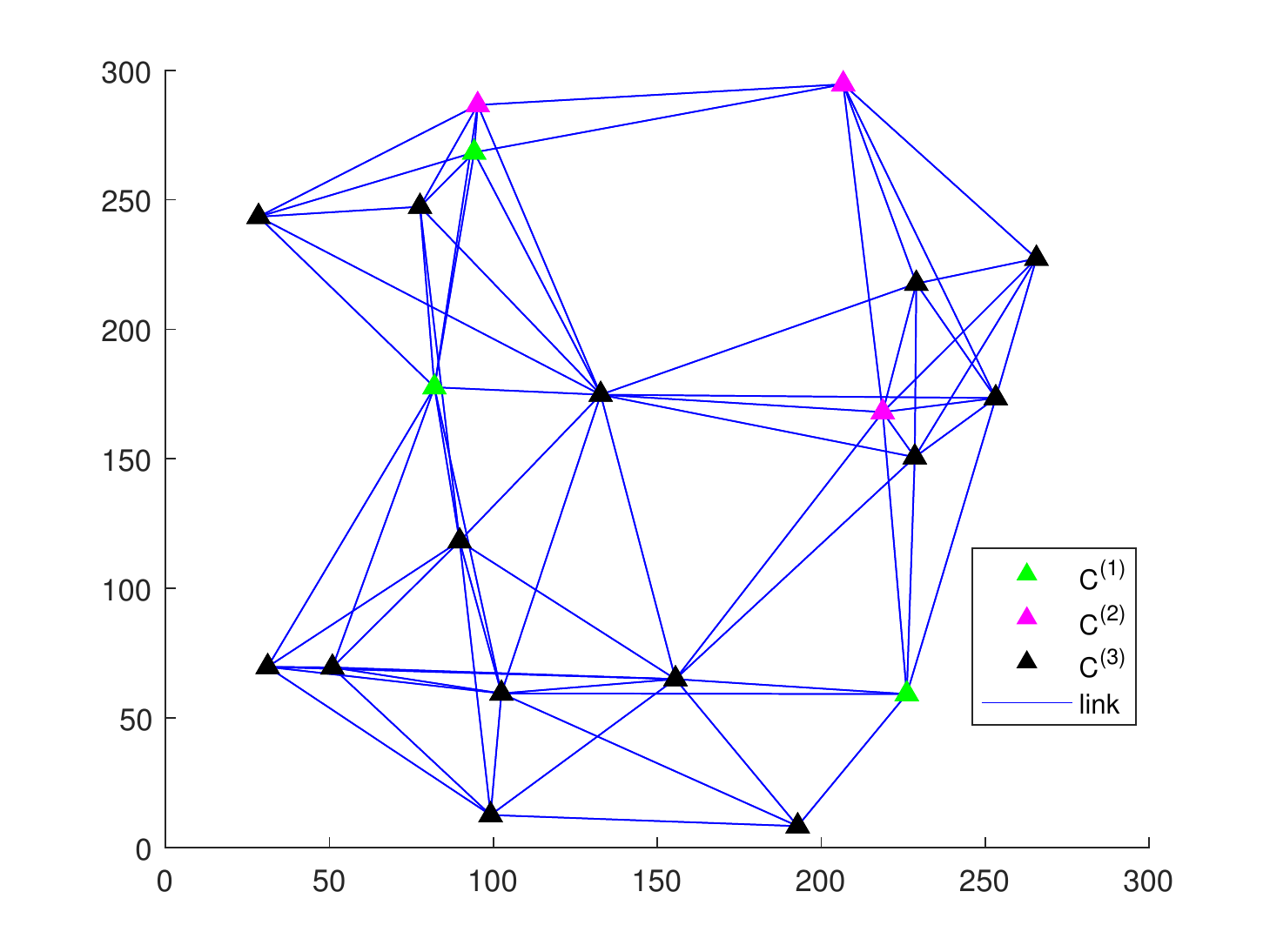}
	\caption{
	Illustration figure for communication topology of the sensor network.}
    \label{topology}	
\end{figure}

The expression of three kinds of observation matrices indicates that each kind of sensor is able to measure some partial states of the system but none of the pair $\left(A,C_i\right)$ is fully observable. Thus, the numerical experiments can fully illustrate the effectiveness of the CMDF algorithm, {as well as} the relationship between the filtering performance and fusion step $L$. Furthermore, the existence of $C^{(3)}$ implies that some sensor nodes in the network do not have the ability to observe the target system. This kind of node is {called a naive node}\cite{battistelli2018distributed}. 

In the experiment, the steady-state performance of the CMDF algorithm is {evaluated by the mean square error (MSE). Using} {Monte Carlo} method, the filtering process is run 200 steps for each simulation and 1000 times in total. With the result {obtained} in Section \ref{sec3}.1, in each simulation, the performance of the CMDF algorithm converges to the steady-state performance. Thus, the estimated value in the steady-state zone {is used} to calculate the mean square error, i.e.,
\begin{equation}\label{mymse}
	MSE_{i,k}=\frac{1}{1000}\sum_{l=1}^{1000}\big\|\hat{x}_{i,\infty}^{\left(l\right)}-x_{\infty}^{(l)}\big\|_{2}^{2},
\end{equation}
where $\hat{x}_{i,\infty}^{\left(l\right)}$ and $x_{\infty}^{(l)}$ denote the estimated state and real state in the $l$-th simulation, respectively, and $\infty$ means {sufficiently large} sampling instant $k$.

 Through changing the value of fusion step $L$, {one can} calculate the steady-state performance of each sensor with different values of fusion step $L$. The performance of CM-based distributed filtering algorithm with different fusion step $L$ is presented in Fig.\ref{partial} and Fig.\ref{total}, where Fig.\ref{partial} illustrates the steady-state performance of mean square error for partial sensors to fully illustrate the exponential convergence and Fig.\ref{total} illustrates the steady-state performance with different $L$ for the whole sensor network. Fig.\ref{partial} and Fig.\ref{total} verify the results of Theorem \ref{thmric} and Theorem \ref{thmcov}.
\begin{figure}
	\centering
	\includegraphics[width=0.4\textwidth]{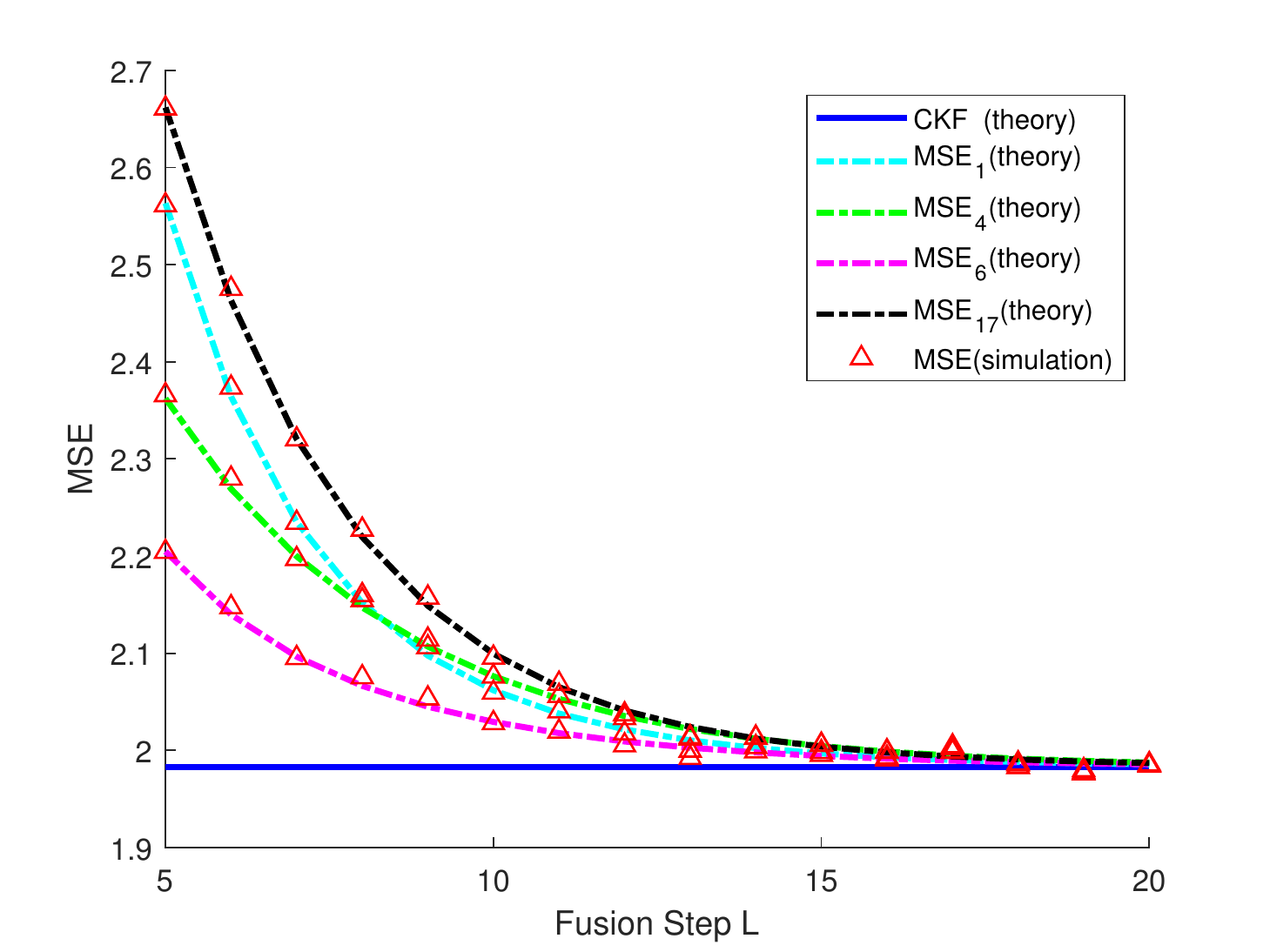}
	\caption{
	Illustration figure for convergence of filtering performance with the increasing of $L$, where CKF(theory) denotes the steady-state MSE of the centralized Kalman filtering algorithm, $\text{MSE}_i$(theory) denotes the steady-state MSE of the $i$-th sensor calculated through \eqref{lyaeq} and MSE(simulation) denotes the steady-state MSE calculated through \eqref{mymse}.}
    \label{partial}
\end{figure}

\begin{figure}
	\centering
	\includegraphics[width=0.4\textwidth]{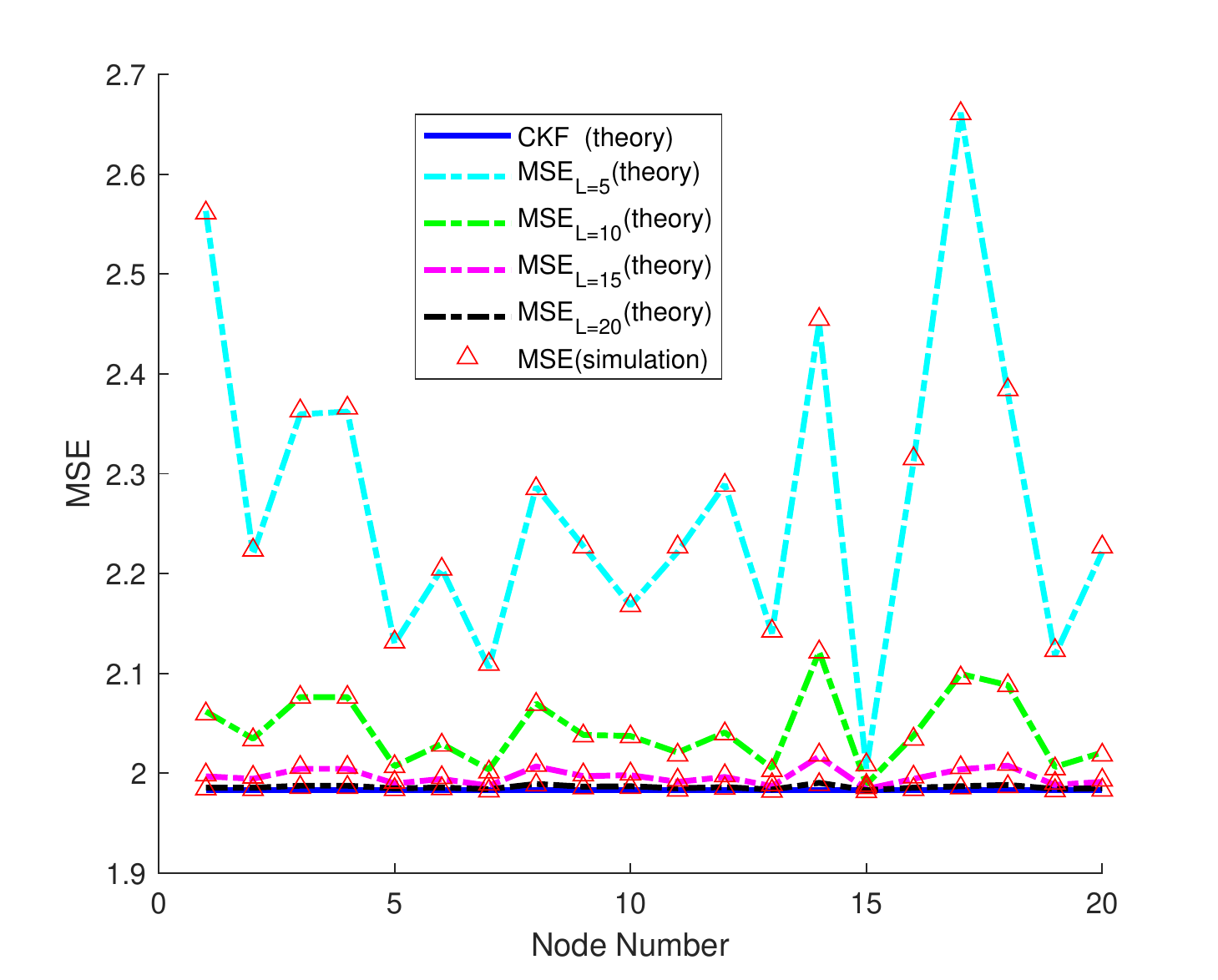}
	\caption{
		Illustration figure for convergence of filtering performance with the increasing of $L$, where $\text{MSE}_{L=k}$(theory) denotes MSE for all sensors with fusion step $L$ equal to $k$.}
    \label{total}
\end{figure}

\begin{remark}
	It is worth mentioning that the filtering performance of the algorithm in \cite{kamal2013information} and \cite{talebi2019distributed} also showed exponential convergence {to the centralized setting with the increasing of $L$}, in their numerical experiments but without proof. In this paper, a rigorous proof of the exponential convergence of the CMDF algorithm {is provided}.
\end{remark}

\section{Conclusion}\label{sec5}
In this paper, the performance of the CMDF algorithm is detailedly investigated. With the derived results on the properties of solutions to DARE, the trade-off between the fusion step and the steady-state performance of the filtering algorithm {is revealed}, including the performances of steady-state matrix iterations and steady-state error covariance matrix. It is shown that with the number of fusion step tending to infinity, both the matrix iteration performance and error covariance performance converge to the centralized optimal ones exponentially.
Compared with the literature, the results established in this paper {theoretically supplement to} the problem of consensus-based distributed filtering. In addition to this, {the results on the} DARE can be utilized in other fields of state estimation. Future work will include a preciser analysis on the exponential bound of $\big\|\tilde{P}_i^{(L)}-P\big\|_2$, {including parameter selection for controlling the steady-state performance of the algorithm, and parameter tuning technique for a better convergence rate. Since the inconsistency analysis of the CM-based distributed filering algorithm is important for performance evaluation, effective methods will be developed for dealing with inconsistency in future work.}

\appendix

\section{Proof of Lemma \ref{lm3}}\label{Pflm3}
	Consider the Schur Decomposition of matrix $A$ as
	$
	A=U^H TU
	$, where $U$ is a unitary matrix and $T$ is an upper-triangular matrix with the diagonal elements being the eigenvalues of $A$. As the singular value is unitary invariant, it is {easy} to obtain  
	$$
	\lambda_i \left(T\right)=\lambda_i \left(A\right),\qquad \big\|T\big\|_2=\big\|A\big\|_2.
	$$ 
	By the definition of 2-norm, one has
	$$
	\big\|A\big\|_2^2=\max\left(\lambda\left(AA^T\right)\right)=\max\left(\lambda\left(TT^T\right)\right).
	$$
	As the maximum eigenvalue of the positive definite matrix $AA^T$ is larger than all the diagonal elements of $AA^T$, one has
	$$
	\big\|A\big\|_2^2\ge \max_{i}\sum_{j=1}^n t_{ij}^2,\qquad\max_{i,j}\left|t_{ij}\right|\leq \big\|A\big\|_2,
	$$
	where $t_{ij}$ is the $\left(i,j\right)$-th element of $T$. Denote $M_T=\max_{i,j}\left|t_{ij}\right|$. {By Corollary 3.15 in \cite{dowler2013bounding}, one has}
	$$
	\begin{aligned}
		\left\|A^k\right\|_2&\leq\sqrt{n}\sum_{j=0}^{n-1}\binom{n-1}{j}\binom{k}{j}M_T^j\rho\left(A\right)^{k-j}\\
		&\leq\sqrt{n}\sum_{j=0}^{n-1}\binom{n-1}{j}\binom{k}{j}\big\|A\big\|_2^j\rho\left(A\right)^{k-j}.
	\end{aligned}
	$$ $\hfill \square$

\section{Proof of Theorem \ref{thmconverge}}\label{Pfconverge}
With Algorithm \ref{alg1}, one obtains
	$$
    \begin{aligned}
	 \tilde{x}_{i,k|k}=&P_{i,k|k}P_{i,k|k}^{-1}x_k-P_{i,k|k}\\
	 \times&\Big(P_{i,k|k-1}^{-1}x_{i,k|k-1}+\sum_{j=1}^{N}l_{ij}^{(L)}C_{j}^{T}R_{j}^{-1}y_j\Big)\\
	 =&\;P_{i,k|k}\Big(P_{i,k|k-1}^{-1}\tilde{x}_{i,k|k-1}-\big(\tilde{C}_i^{(L)}\big)^T\big(\tilde{R}_i^{(L)}\big)^{-1}\tilde{v}_{i,k}\Big),\\
	 \tilde{P}_{i,k|k}=&P_{i,k|k}P_{i,k|k-1}^{-1}\tilde{P}_{i,k|k-1}P_{i,k|k-1}^{-1}P_{i,k|k}\\
	 +&P_{i,k|k}\big(\tilde{C}_i^{(L)}\big)^T\big(\tilde{R}_i^{(L)}\big)^{-1}\bar{R}_i^{(L)}\big(\tilde{R}_i^{(L)}\big)^{-1}\tilde{C}_i^{(L)}P_{i,k|k}
    \end{aligned}
    $$
    and
    $$
     \begin{aligned}
	 \tilde{x}_{i,k+1|k}=&A\tilde{x}_{i,k|k}+\omega_{k},\\
	 \tilde{P}_{i,k+1|k}=&A\tilde{P}_{i,k|k}A^T+Q,
      \end{aligned}
    $$
    where $\tilde{v}_{i,k}=\Big[sign\big(l_{i1}^{(L)}\big) v_{1,k}^T,\cdots,sign\big(l_{iN}^{(L)}\big)v_{N,k}^T\Big]^T$.
    
    {It follows that}
    $$
    \begin{aligned}
	    AP_{i,k|k}P_{i,k|k-1}^{-1}&=A-AP_{i,k|k-1}\big(\tilde{C}_i^{(L)}\big)^{T}\\
	    \times&\Big(\tilde{C}_i^{(L)}P_{i,k|k-1}\big(\tilde{C}_i^{(L)}\big)^T+\tilde{R}_i^{(L)}\Big)^{-1}\tilde{C}_i^{(L)},
    \end{aligned}
    $$
    and
    $$
    \begin{aligned}
    	P_{i,k|k}\big(\tilde{C}_i^{(L)}\big)^{T}&\big(\tilde{R}_{i}^{(L)}\big)^{-1}=P_{i,k|k-1}\big(\tilde{C}_i^{(L)}\big)^{T}\\
    	&\quad\;\;\times\Big(\tilde{C}_i^{(L)}P_{i,k|k-1}\big(\tilde{C}_i^{(L)}\big)^{T}+\tilde{R}_{i}^{(L)}\Big)^{-1}.\qquad\quad
    \end{aligned}
    $$
	As $k$ tends to infinity, these two matrices will converge to the steady-state forms, i.e.,
    $$
    \begin{aligned}
	    \lim_{k\to\infty}&AP_{i,k|k}P_{i,k|k-1}^{-1}=\tilde{A}_{P_i^{(L)}},\\
	    \lim_{k\to\infty}&AP_{i,k|k}\tilde{C}_{i}^T\big(\tilde{R}_{i}^{(L)}\big)^{-1}=K_{P_i^{(L)}}.
    \end{aligned}
    $$  
  {With $P_i^{(L)}$ being the solution to DARE, the above matrix $\tilde{A}_{P_i^{(L)}}$ is Schur stable}, i.e., the spectral radius of $\tilde{A}_{P_i^{(L)}}$ is less than $1$ \cite{kailath2000linear}.
    Thus, with Theorem 1 in \cite{cattivelli2010diffusion}, the steady-state performance of $\tilde{P}_{i,k|k-1}$ will also converge, i.e., $\lim_{k\to\infty}\tilde{P}_{i,k|k-1}=\tilde{P}_i^{(L)}$ and the convergent solution  satisfies the DLE \eqref{lyaeq}.$\hfill \square$

\section{Proof of Lemma \ref{lmstable}}\label{Pflmstable}
Consider the DARE
	$$
     P=APA^T+Q-APC^T\left(CPC^T+R\right)^{-1}CPA.
    $$
	Following the similar idea in \cite{kailath2000linear}, the above equation is rewritten as
	\begin{equation}\label{eqstable}
		P = \tilde{A}_P P\tilde{A}_P^T+Q+K_PRK_P^T,
	\end{equation}
	where $K_P = APC^T\left(CPC^T+R\right)^{-1}$. Consider the left eigenvector $x$ of $\tilde{A}_P$ with corresponding eigenvalue $\lambda$, i.e., $x\tilde{A}_P=\lambda x$. Pre- and post-multiplying the DARE (\ref{eqstable}) with $x$ and $x^H$, respectively, one has
	$$
	\begin{aligned}
		xPx^H&=\left|\lambda\right|^2xPx^H+xQx^H+xK_PRK_P^Tx^H\\
		&\ge\left|\lambda\right|^2xPx^H+xQx^H.
	\end{aligned}
	$$
	The inequality holds due to the fact that the matrix $K_PRK_P^T$ is positive semi-definite. {Dividing} both side of the inequality with $xPx^H$, one has
	$$
	\left|\lambda\right|^2\leq 1-\frac{xQx^H}{xPx^H}\leq 1-\min_{x}\frac{xQx^H}{xPx^H}\leq 1-\frac{\lambda_{min}\left(Q\right)}{\lambda_{max}\left(P\right)},
	$$
	thus the inequality of the spectral radius is proved. 
	
	For the inequality of the 2-norm, consider the singular value decomposition $\tilde{A}_P=U\Sigma V^T$, where $\Sigma=diag\left\lbrace \sigma_1,\dots,\sigma_n\right\rbrace$, $\sigma_1\ge\dots\ge\sigma_n\ge0$. Let $x=\left(1,0,\dots,0\right)U^T$. Pre- and post-multiplying the DARE (\ref{eqstable}) with $x$ and $x^T$, respectively, one has
	$$
	\begin{aligned}
		xPx^T=\sigma_1^2\left(V^TPV\right)_{1,1}+xQx^T+xKRK^Tx^T,
	\end{aligned}
	$$
	where $\big(V^TPV\big)_{1,1}$ is the $(1,1)$-th element of matrix $V^TPV$. Note that $\left(V^TPV\right)_{1,1}\ge \lambda_{min}\left(V^TPV\right)=\lambda_{min}\left(P\right)$. Thus, one has
	$$
	\sigma_1^2\leq \frac{\lambda_{max}\left(P\right)}{\lambda_{min}\left(P\right)}\leq \frac{\lambda_{max}\left(P\right)}{\lambda_{min}\left(Q\right)},
	$$
	where the last inequality {holds because} $P\ge Q$. $\hfill \square$
\section{Proof of Lemma \ref{lmdifference}}\label{Pflmdifference}
    As $P_1$ and $P_2$ satisfy the DARE \eqref{twodare}, one has
	\begin{equation}\label{minus}
		\begin{aligned}
			P_1-P_2=&A\left(P_1-P_2\right)A^T\\
			&-K_{P_1}C\left(P_1-P_2\right)A^T-A\left(P_1-P_2\right)C^TK_{P_2}^T\\
			&-K_{P_1}CP_2A^T+AP_1C^TK_{P_2}^T.\\
		\end{aligned}
	\end{equation}
	Note that 
	$$
	\begin{aligned}
		&AP_1C^TK_{P_2}^T-K_{P_1}CP_2A^T\\
		=&K_{P_1}\left(CP_1C^T+R_1\right)K_{P_2}-K_{P_1}\left(CP_2C^T+R_2\right)K_{P_2}\\
		=&K_{P_1}C\left(P_1-P_2\right)C^TK_{P_2}^T+K_{P_1}\left(R_1-R_2\right)K_{P_2}^T.
	\end{aligned}
	$$
	Thus, the original equation \eqref{minus} can be rewritten as
	$$
	\begin{aligned}
		P_1-P_2=&\left(A-K_{P_1}C\right)\left(P_1-P_2\right)\big(A-K_{P_2}C\big)^T\\
		&+K_{P_1}\left(R_1-R_2\right)K_{P_2}^T\\
		=&\tilde{A}_{P_1}\left(P_1-P_2\right)\tilde{A}_{P_2}^T+K_{P_1}\left(R_1-R_2\right)K_{P_2}^T.
	\end{aligned}
	$$
{Performing} the above iteration for $k$ times, one has
	$$
	\begin{aligned}
		P_1-P_2=&\tilde{A}_{P_1}^k\left(P_1-P_2\right)\big(\tilde{A}_{P_2}^T\big)^k\\
		&+\sum_{i=0}^{k-1}\tilde{A}_{P_1}^iK_{P_1}\left(R_1-R_2\right)K_{P_2}^T\big(\tilde{A}_{P_2}^T\big)^i.
	\end{aligned}
	$$
	Due to the fact that $\tilde{A}_{P_1}=A-K_{P_1}C$ and $\tilde{A}_{P_2}=A-K_{P_2}C$ are Schur stable, {by performing the iteration for infinitely many times,} one has
	$$
	\lim_{k\to\infty}\tilde{A}_{P_1}^k\left(P_1-P_2\right)\big(\tilde{A}_{P_2}^T\big)^k={\bf O},
	$$
	and
	\begin{equation}\label{infisum}	
			P_1-P_2=\sum_{i=0}^{\infty}\tilde{A}_{P_1}^iK_{P_1}\left(R_1-R_2\right)K_{P_2}^T\big(\tilde{A}_{P_2}^T\big)^i.
	\end{equation}
	With the Kalman identical equation, one has
	$$K_{P_1}=A\bar{P}_1C^TR_1^{-1}\qquad K_{P_2}=A\bar{P}_2C^TR_2^{-1}.$$ 
	Therefore, the equation (\ref{infisum}) can be rewritten as
	$$
	P_1-P_2=\sum_{i=0}^{\infty}\tilde{A}_{P_1}^iA\bar{P}_1C^T\left(R_2^{-1}-R_1^{-1}\right)C\bar{P}_2A^T\big(\tilde{A}_{P_2}^T\big)^i.
	$$ $\hfill \square$

\section{Proof of Theorem \ref{thmric}}\label{Pfthmric}
	{First,} analyze the property of $\big\|\big(\tilde{R}_i^{(L)}\big)^{-1}-R^{-1}\big\|_2$.
	
	{Using} the designed algorithm, one can obtain the following relationship:
	$$
	\big(\tilde{R}_{i}^{(L)}\big)^{-1}=Ndiag\big(l_{i1}^{(L)}I_{n_1},\cdots,l_{iN}^{(L)}I_{n_N}\big)R^{-1}.
	$$
 {By} Lemma \ref{lm1}, the matrix $\mathcal{L}^{L}$ converges to $\frac{1}{N}\textbf{1}_N\textbf{1}_N^T$ exponentially {with increasing $L$}, which indicates that each term $Nl_{ij}^{(L)}$ converges to $1$ exponentially. Thus, there exist two constants $M_0>0$ and $0<q<1$ such that, for any $L\ge d$ and $i\in\mathcal{V}$, 
	$$
	\big\|\big(\tilde{R}_{i}^{(L)}\big)^{-1}-R^{-1}\big\|_2<M_0q^L.
	$$
	The main idea of the following proof is to find {a uniform bound on} $\sum_{k=0}^{\infty}\big\|\tilde{A}_{P}^k\big\|_2\big\|\tilde{A}_{P_i^{(L)}}^k\big\|_2$ for $\tilde{R}_i^{(L)}$, $\forall i\in\mathcal{V}, L\ge d$. 
	With Lemma \ref{lm3}, one can obtain
	\begin{small}
	$$
	\Big\|\tilde{A}_{P_i^{(L)}}^k\Big\|_2\leq\sqrt{n}\sum_{j=0}^{n-1}\binom{n-1}{j}\binom{k}{j}\Big\|\tilde{A}_{P_{i}^{(L)}}\Big\|_2^j\rho\Big(\tilde{A}_{P_{i}^{(L)}}\Big)^{k-j},
	$$
	\end{small}

	{where $n$ is the dimension of the system. Choose a sufficiently large $\mathcal{R}$}, such that $\tilde{R}_{i}^{(L)}\leq\mathcal{R}$ for all $i\in\mathcal{V}$ and $L\ge d$. Denote $\mathcal{P}=dare\left(A,C,Q,\mathcal{R}\right)$, 
	With Lemma \ref{lmmono} and Lemma \ref{lmstable}, one can obtain 
	$$
	 \begin{aligned}
		&\big\|\tilde{A}_{P_{i}^{(L)}}\big\|_2\leq \sqrt{\frac{\lambda_{max}\big(P_{i}^{(L)}\big)}{\lambda_{min}\left(Q\right)}}\leq\sqrt{\frac{\lambda_{max}\left(\mathcal{P}\right)}{\lambda_{min}\left(Q\right)}}\triangleq\sigma_{1\mathcal{P}},\\
		&\rho\big(\tilde{A}_{P_{i}^{(L)}}\big)\leq \sqrt{1-\frac{\lambda_{min}\left(Q\right)}{\lambda_{max}\big(P_{i}^{(L)}\big)}}\leq\sqrt{1-\frac{\lambda_{min}\left(Q\right)}{\lambda_{max}\left(\mathcal{P}\right)}}\triangleq \rho_{\mathcal{P}}.
	 \end{aligned}
	$$
	Thus, one has
	$$
	\begin{aligned}
		\sum_{k=0}^{\infty}\Big\|\tilde{A}_{P_i^{(L)}}^k\Big\|_2\leq&\sum_{k=0}^{\infty}\sqrt{n}\sum_{j=0}^{n-1}\binom{n-1}{j}\binom{k}{j}\sigma_{1\mathcal{P}}^j\rho_{\mathcal{P}}^{k-j}\\
		\leq&\sum_{k=0}^{\infty}\sqrt{n}\left(n\max_{j}\binom{n-1}{j}\sigma_{1\mathcal{P}}^j\right)k^n\rho_{\mathcal{P}}^{k-n}
	\end{aligned}
	$$
	for all $i\in\mathcal{V}$ and $L\ge d$. It is {easy} to verify the convergence of the infinite sum $\sum_{k=0}^{\infty}\big\|\tilde{A}_{P_i^{(L)}}^k\big\|_2$ with $\rho_{\mathcal{P}}<1$. {Based on} the fact that the 2-norm $\big\|\tilde{A}_{P}^k\big\|_2$ is uniformly bounded for all $k$, one can find a real number $M_0^{'}$ such that
	$$
	\sum_{k=0}^{\infty}\Big\|\tilde{A}_{P}^k\Big\|_2\Big\|\tilde{A}_{P_i^{(L)}}^k\Big\|_2\leq M_0^{'},\;\;\forall i\in\mathcal{V},\forall L\ge d.
	$$
	With the above analysis, the theorem is proved.$\hfill \square$

\section{Proof of Theorem \ref{thmcov}}\label{Pfthmcov}
The matrices $P_{i}^{(L)}$ and $\tilde{P}_i^{(L)}$ satisfy the following equations:
$$
\begin{aligned}
	&P_{i}^{(L)}=\tilde{A}_{P_i^{(L)}}P_{i}^{(L)}\tilde{A}_{P_i^{(L)}}^T+Q+K_{P_{i}^{(L)}}\tilde{R}_i^{(L)}K_{P_{i}^{{(L)}}}^T,\\
	&\tilde{P}_i^{(L)}=\tilde{A}_{P_i^{(L)}}\tilde{P}_i^{(L)}\tilde{A}_{P_i^{(L)}}^T+Q+K_{P_{i}^{(L)}}RK_{P_{i}^{{(L)}}}^T.
\end{aligned}
$$
Then, one can obtain

$$
\begin{aligned}
	\tilde{P}_{i}^{(L)}-P_{i}^{(L)}=&\tilde{A}_{P_i^{(L)}}\big(\tilde{P}_{i}^{(L)}-P_{i}^{(L)}\big)\tilde{A}_{P_i^{(L)}}^T\\
	&\qquad\qquad+K_{P_{i}^{(L)}}\big(R-\tilde{R}_{i}^{(L)}\big)K_{P_{i}^{{(L)}}}^T,
\end{aligned}
$$
and
\begin{equation}\label{sumCov}
	\tilde{P}_{i}^{(L)}-P_{i}^{(L)}=\sum_{k=0}^{\infty}\tilde{A}_{P_i^{(L)}}^k K_{P_{i}^{(L)}}\big(R-\tilde{R}_{i}^{(L)}\big)K_{P_{i}^{{(L)}}}^T\big(\tilde{A}_{P_i^{(L)}}^T\big)^k.
\end{equation}

{Using} $K_{P_{i}^{(L)}}=A\bar{P}_i^{(L)}C^T\big(\tilde{R}_{i}^{(L)}\big)^{-1}$, one can rewrite \eqref{sumCov} as
$$
\begin{aligned}
	\tilde{P}_{i}^{(L)}-P_{i}^{(L)}=&\sum_{k=0}^{\infty}\tilde{A}_{P_i^{(L)}}^k A\bar{P}_i^{(L)}C^T\big(\tilde{R}_{i}^{(L)}\big)^{-1}\\
	&\times\big(R-\tilde{R}_{i}^{(L)}\big)\big(\tilde{R}_{i}^{(L)}\big)^{-1}C\bar{P}_i^{(L)}A^T\big(\tilde{A}_{P_i^{(L)}}^T\big)^k.
\end{aligned}
$$
{Thus, it follows that}
	$$
	\begin{aligned}
		\Big\|\tilde{P}_i^{(L)}-P_i^{\left(L\right)}\Big\|_2\leq& \Big\|\big(\tilde{R}_{i}^{(L)}\big)^{-1}\big(R-\tilde{R}_{i}^{(L)}\big)\big(\tilde{R}_{i}^{(L)}\big)^{-1}\Big\|_2\\
		&\times\big\|\bar{P}_i^{(L)}\big\|_2^2\big\|C\big\|_2^2\big\|A\big\|_2^2\sum_{k=0}^{\infty}\big\|\tilde{A}_{P_i^{(L)}}^k\big\|_2^2.
	\end{aligned}
	$$
	Similarly to the proof in Theorem \ref{thmric}, one can find a uniform bound as $\big\|\bar{P}_i^{(L)}\big\|_2^2\sum_{k=0}^{\infty}\big\|\tilde{A}_{P_i^{(L)}}^k\big\|_2^2\leq M_2^{'},\forall i\in\mathcal{V},\forall L\ge d$ with the $\mathcal{R},\mathcal{P}$ determined in the proof of Theorem 2. {Thus, it suffices} to discuss the term $\big(\tilde{R}_{i}^{(L)}\big)^{-1}\big(R-\tilde{R}_{i}^{(L)}\big)\big(\tilde{R}_{i}^{(L)}\big)^{-1}$ to {show} the convergence with increasing $L$.

	Both matrices $R$ and $\tilde{R}_{i}^{(L)}$ are block diagonal. Thus, one can obtain 
	\begin{small}
	$$
	\begin{aligned}
		\big(\tilde{R}_{i}^{(L)}\big)^{-1}\big(R-\tilde{R}_{i}^{(L)}\big)\big(\tilde{R}_{i}^{(L)}\big)^{-1}
		=diag\big(\tilde{l}_{i1}^{(L)}I_{n_1},\cdots,\tilde{l}_{i1}^{(L)}I_{n_N}\big)R^{-1},
	\end{aligned}
	$$
	\end{small}

	where $\tilde{l}_{ij}^{(L)}=N^2\big(l_{ij}^{(L)}\big)^2-Nl_{ij}^{(L)}$. One can rewrite the expression of $\tilde{l}_{ij}^{(L)}$ as 
	$$
	\tilde{l}_{ij}^{(L)}=\big(Nl_{ij}^{(L)}-1\big)^2+\big(Nl_{ij}^{(L)}-1\big).
	$$
	With Lemma \ref{lm1}, the term $Nl_{ij}^{(L)}-1$ converges to $0$ exponentially with increasing $L$, i.e., there exists $M_3>0$, $0<q_3<1$, such that $\forall i,j$, one has $\big|Nl_{ij}^{(L)}-1\big|\leq M_3q_3^L$. {For $L\ge log_{\frac{1}{q_3}}M_3$, one has} $\big(Nl_{ij}^{(L)}-1\big)^2<\big|Nl_{ij}^{(L)}-1\big|$, {therefore there exist two parameters} $M_2^{''}\ge 0$ and $0<q_2<1$, such that
	$$
	\begin{aligned}
		\Big\|\big(\tilde{R}_{i}^{(L)}\big)^{-1}\big(R-\tilde{R}_{i}^{(L)}\big)\big(\tilde{R}_{i}^{(L)}\big)^{-1}\big\|_2\leq M_2^{''}q_2^{L}
	\end{aligned}
	$$ 
	for all $i\in\mathcal{V}$ and $L\ge d$. 
	The theorem is proved.$\hfill \square$

\bibliographystyle{IEEEtran}
\bibliography{ref}

\end{document}